\documentclass[lineno]{biometrika}

\usepackage{amsmath}

\RequirePackage[colorlinks,citecolor=blue,urlcolor=blue]{hyperref}

\usepackage{newtxtext}
\usepackage[subscriptcorrection]{newtxmath}

\usepackage{dsfont}
\usepackage{float}
\usepackage{caption}
\usepackage[plain,noend]{algorithm2e}

\makeatletter
\renewcommand{\algocf@captiontext}[2]{#1\algocf@typo. \AlCapFnt{}#2}

\def\@algocf@capt@plain{top}
\renewcommand{\algocf@makecaption}[2]{%
  \addtolength{\hsize}{\algomargin}%
  \sbox\@tempboxa{\algocf@captiontext{#1}{#2}}%
  \ifdim\wd\@tempboxa >\hsize%
    \hskip .5\algomargin%
    \parbox[t]{\hsize}{\algocf@captiontext{#1}{#2}}%
  \else%
    \global\@minipagefalse%
    \hbox to\hsize{\box\@tempboxa}%
  \fi%
  \addtolength{\hsize}{-\algomargin}%
}
\makeatother

\def\E{E}
\def\P{{\rm pr}}

\def\Sset{\mathcal{S}}
\def\M{M_\text{int}}
\def\Ps{\mathcal{P}_\Sset}
\def\Bset{\mathcal{B}}
\def\Pb{\mathcal{P}_\Bset}

\def\data{\mathcal{D}_n}
\def\test{\mathcal{T}}

\newfloat{algorithm}{H}{loa}
\floatstyle{plaintop}
\floatname{algorithm}{\textit{Algorithm}}
\restylefloat{algorithm}

\usepackage{hyperref}

\begin{document}

\nolinenumbers

\jname{preprint}
\copyrightinfo{}

\markboth{N. Williams et~al.}{Controlling random seed stability via bagging}

\title{Improving reproducibility by controlling random seed stability in machine learning based estimation via bagging}

\author{N. WILLIAMS}
\affil{Division of Biostatistics, University of California, Berkeley,\\ California, U.S.A.
\email{nicholaswilliams@berkeley.edu}}

\author{\and A. SCHULER}
\affil{Division of Biostatistics, University of California, Berkeley,\\ California, U.S.A. \email{alejandro.schuler@berkeley.edu}}

\maketitle

\begin{abstract}
Predictions from machine learning algorithms can vary across random seeds, inducing instability in downstream debiased machine learning estimators. We formalize random seed stability via a concentration condition and prove that subbagging guarantees stability for any bounded-outcome regression algorithm. We introduce a new cross-fitting procedure, adaptive cross-bagging, which simultaneously eliminates seed dependence from both nuisance estimation and sample splitting in debiased machine learning. Numerical experiments confirm that the method achieves the targeted level of stability whereas alternatives do not. Our method incurs a small computational penalty relative to standard practice whereas alternative methods incur large penalties.
\end{abstract}


\section{Introduction}

The use of machine learning in semiparametric estimation is becoming more common. However, cautionary warnings have appeared that such estimators may suffer from ``random seed dependence'': running the same estimator on the same data yields a different result each time \citep{schader2024don,naimi2024pseudo}. To illustrate, consider a trial statistician who pre-specifies a machine learning method for covariate adjustment to improve power \citep{williams2022optimising}. The drug shows no effect. Knowing the estimator has stochastic dependence, the statistician sets a new seed, reruns the analysis, finds a significant effect, and submits those results. Clearly, such a scenario is a form of ``p-hacking'' and should be prevented. 

Prior study of random seed stability has, to our knowledge, been limited to Monte Carlo experimentation under arbitrary data generating processes. While informative, such results have limited generalizability and offer no theoretical guarantees. That being said, the broad conclusion across this literature is consistent: aggregate over enough seeds \citep{chernozhukov2018double,schader2024don,naimi2024pseudo} and use enough folds \citep{williams2025re,zivich2024commentary}---but without ever addressing what constitutes enough. There are two sources of random seed dependence in causal machine learning: (i) the stochastic nature of the machine learning algorithms used to estimate nuisance parameters, and (ii) randomness induced by sample splitting for cross-fitting \citep{naimi2024pseudo}. The primary focus of this paper is the former, and we show that bagging can eliminate stochastic dependence from the algorithms themselves. However, because bagging operates through resampling, it also provides a natural mechanism for cross-fitting without explicit sample splitting, which addresses the second source as well.

In this paper, we study the problem of random seed dependence through the lens of algorithmic stability. We first introduce a formal definition of random seed stability in the form of a concentration condition on the probability that the difference in predictions from the same supervised learning algorithm trained on the same data using different seeds exceeds some threshold. We then build on previous work that shows that bootstrap aggregating (bagging) provides algorithmic stability guarantees \citep{soloff2024bagging}. Our contributions are as follows:
\begin{enumerate}
\item we prove that with enough bags the subbagged version of any machine learning algorithm satisfies our definition of being seed stable, 
\item we prove that a generic real-valued function of bagged predictions is itself seed stable under regularity conditions, and
\item we introduce a new cross-fitting algorithm, adaptive cross-bagging, that guarantees seed stability.
\end{enumerate}

\section{Defining random seed stability}

Let $z_i = (x_i,y_i)$, where $x_i \in \mathcal X \subseteq \mathds{R}^d$ and $y_i \in \mathds{R}$. Let $z_1, z_2, \dots$ be a deterministic sequence, and let $\mathcal{D}_n = \{z_i\}_{i=1}^n$ so that $\mathcal{D}_n \subset \mathcal{D}_{n+1}$; $\mathcal{D}_n$ is a fixed training dataset of $n$ observations. Because $\data$ is finite and fixed, $\{y_i\}_{i=1}^n$ is bounded and, without loss of generality, we rescale $y_i\in [0,1]$ via min-max scaling. Let $\test \subset \mathcal X$ with $|\test| = k$ be a held-out set of $k$ predictor values. Let $\Sset = \left\{1, \dots, \M\right\}$ where $\M$ is the largest integer that can be represented on a machine; define $\Ps := \text{Unif}(\Sset)$. Let $\mathfrak F$ denote a generic supervised machine learning algorithm (referred to as a regression algorithm throughout) and $\mathfrak{F}(\data, s) : \mathcal X \rightarrow [0,1]$ be a realized prediction function, where $s \in \Sset$ is a random seed. Throughout, a subscript on $E$ or $\P$ indicates expectation or probability over that variable alone, with all other quantities fixed (e.g., $\P_s$ is taken with respect to $\Ps$). We let $\mathds{1}[A]$ be the indicator function that returns one when event $A$ occurs and zero otherwise, and $\mathbf{1}_V$ be a vector of ones with size $V$. 

\begin{definition}\label{def:stability}
A regression algorithm $\mathfrak F$ is random seed $(\epsilon,\delta)$-stable on training data $\data$ and testing set $\test$ if for all values $x \in \test$
\[
\P_{s,s'}\left(\left| \mathfrak{F}(\data, s)(x) - \mathfrak{F}(\data, s')(x)\right| \geq \epsilon \right) \leq \delta.
\]
\end{definition}

Definition~\ref{def:stability} aligns with practical application: an analyst sets a seed $s$, trains the machine learning algorithm, and predicts on a new data point(s). Furthermore, Definition~\ref{def:stability} mirrors classical notions of algorithmic stability \citep{shalev2010learnability}, which bound the effect of perturbing the training data; here, the perturbation is instead to the seed. If the variance of predictions across seeds is small enough then the algorithm will be seed stable. 
\begin{lemma}\label{lem:basestability}
Assume $\mathfrak{F}$ is a regression algorithm such that for any dataset $\mathcal{D}$ and seed $s$, the fitted function $\mathfrak{F}(\mathcal{D}, s)$ takes values in $[0,1]$. Define the idealized point-wise prediction function over all seeds as
\begin{equation*}
    {\mathfrak F}^\infty(\data)(x) = \E_s\left[\mathfrak{F}(\data, s)(x)\right],
\end{equation*}
and the point-wise prediction variance over seeds is
\begin{equation*}\label{eq:basemse}
    \sigma^2(x; \data) = \E_s\left[\left(\mathfrak{F}(\data, s)(x) -  {\mathfrak F}^\infty(\data)(x) \right)^2 \right].
\end{equation*}
Then regression algorithm $\mathfrak F$ is random seed $(\epsilon,\delta)$-stable on training data $\data$ and test set $\test$ whenever \[\max_{x \in \test} \sigma^2(x; \data) \le \frac{\epsilon}{4} \left(\frac{\epsilon}{\log(2 k/\delta)} - \frac{2}{3}\right).\]
\end{lemma}

Lemma~\ref{lem:basestability} tells us that seed stability reduces to controlling $\sigma^2(x;\data)$: if the variance of predictions over seeds is sufficiently small, stability is guaranteed. If we suppose that $\sigma^2(x;\data) \rightarrow 0$ as $n \rightarrow \infty$ point-wise over $x$, then the bound in Lemma~\ref{lem:basestability} is eventually satisfied for any $(\epsilon,\delta)$ pair. Therefore any sensible regression algorithm should be seed $(\epsilon,\delta)$-stable on prediction set $\test$ given the training data $\data$ is large enough. However, the rate at which $\sigma^2(x;\data)$ decays at any test point $x$ depends on the algorithm and the training data, and will generally be unknown. In the following sections, we instead show that any bounded regression algorithm can be made seed stable with bagging. 

\section{Bagging and random seed stability}

Bagging is an ensemble method that broadly refers to the practice of repeatedly resampling training data, training regression algorithm $\mathfrak{F}$ on each resample (called ``bags''), and averaging the models \citep{breiman1996bagging,breiman1996heuristics}. Bagging was originally implemented using the nonparametric bootstrap \citep{efron1979bootstrap} by drawing $m = n$ resamples with replacement. We focus on the subbagging variant \citep{andonova2002simple}, which is more computationally tractable and guarantees a fixed number of held-out observations which we will make use of in \S\ref{sec:cmle}. However, our results hold for other bagging variants as well. 

\subsection{Subbagging controls random seed stability}

Let $\rho \in (0,1)$ be a fixed constant, $m = \lfloor \rho n \rfloor$, and $B_m^j = \{i_1, \dots, i_m \} \subset [n]$ where $1 \leq m \leq n$ denote an $m$-size index subsample constructed without replacement. There are then $K = C(n, m)$ possible subsamples; denote the collection of all the $m$-size subsamples as $\Bset_m = \{B_m^1, \dots,  B_m^K\}$; define $\Pb = \text{Unif}(\Bset_m)$. Let $\data(B_m^j)$ be the sub-dataset consisting of the observations of $\data$ indexed by $B_m^j$. 
We then define $\mathfrak F(\data(B_m^j), s)$ as the algorithm trained on the $B_m^j$ subsample and seed $s$. 
Let $\Ps \otimes \Pb$ be the product measure induced by $\Sset \times \Bset$. For $V < K$, define $\zeta = \left\{(s_1, B_m^1), \dots, (s_V, B_m^V)\right\}$ as the set of $V$ independent seed-and-bag pairs each drawn from $\Ps \otimes \Pb$. A subbagged realization of algorithm $\mathfrak F$ is then given by
$$
\bar{\mathfrak F}(\data, \rho, \zeta)(x) = \frac{1}{V} \sum_{v=1}^V \mathfrak{F}(\data(B_m^v), s_v)(x).
$$

\begin{theorem}\label{th:stability}
Assume $\mathfrak{F}$ is a regression algorithm such that for any dataset $\mathcal{D}$ and seed $s$, the fitted function $\mathfrak{F}(\mathcal{D}, s)$ takes values in $[0,1]$. Define the idealized subbagged point-wise prediction function over all bags and seeds as
\[
\bar{\mathfrak F}^\infty(\data, \rho)(x) = E_{(s,B_m)} \left[\mathfrak{F}(\data(B_m), s)(x)\right],
\]
and the point-wise prediction variance over bags and seeds as
\begin{equation*}\label{eq:bagmse}
\nu^2(x; \data,\rho)  = \E_{(s,B_m)}\left[\left(\mathfrak{F}(\data(B_m), s)(x) - \bar{\mathfrak F}^\infty(\data, \rho)(x)\right)^2 \right].
\end{equation*}
Since the $V$ bags are drawn independently, the variance of the subbagged average prediction is $\nu^2(x;\data,\rho) / V$. Then a subbagged version of regression algorithm $\mathfrak F$ is random seed $(\epsilon,\delta)$-stable on training data $\data$ and test set $\test$ when

\[V \ge  \max_{x\in \test}\frac{\log(2k/\delta)}{\epsilon^2}\left(4\,\nu^2(x;\data,\rho) + \frac{2}{3}\,\epsilon\right).\]
\end{theorem}
Similar to Lemma~\ref{lem:basestability}, Theorem~\ref{th:stability} is derived from an application of Bernstein's inequality. The key insight is that $V$ is a free parameter so that any bounded regression algorithm, regardless of its seed stability, can be made seed $(\epsilon,\delta)$-stable by choosing $V$ sufficiently large. By Popoviciu's inequality \citep{popoviciu1935equations}, since $y_i \in [0,1]$, $\nu^2(x;\data,\rho) \leq 1/4$; therefore the worst-case upper bound on the number of bags required for seed $(\epsilon,\delta)$-stability can always be calculated for a regression algorithm bounded in $[0,1]$. However, if we again assume that $\nu^2(x;\data,\rho)  \rightarrow 0$ as $n \rightarrow \infty$ point-wise then the number of bags required for $(\epsilon,\delta)$-stability will decrease with larger training data. 

\section{Implications for debiased machine learning}\label{sec:cmle}

Presume we are now interested in a scalar mapping of $p$ different estimated prediction functions. Let $\eta_s = \{\eta_{s,l}\}_{l=1}^p$, where $\eta_{s,l} = \{\mathfrak F_l (\data, s)(x_i) : z_i \in \data\}$ is a vector of $n$ predictions from the $l$th prediction function generated under seed $s$. Here, unlike in previous sections, predictions are evaluated on the training data such that $\test = \data$. Let $\Psi(\data, \eta_s) \in \mathds R$ be a debiased machine learning estimator (e.g., augmented inverse probability weighting [AIPW], targeted minimum loss-based estimation [TMLE], or double machine learning [DML]) \citep{robins1994estimation,vanderLaan2006,chernozhukov2018double} that maps the data $\data$ and nuisance predictions $\eta_s$ to a real-valued estimate of a target parameter $\psi_0$. 
\begin{definition}
A scalar-valued function $\Psi(\data,\eta)$ is random seed $(\epsilon, \delta)$-stable on $\data$ if
\[
\P_{s,s'}\left(\left|\Psi(\data, \eta_s) - \Psi(\data, \eta_{s'})\right| \geq \epsilon\right) \leq \delta.
\]
\end{definition}
Consider the AIPW estimator for the average treatment effect (ATE). 
Here $z_i = (x_i, a_i, y_i)$, where $a_i \in \{0,1\}$ is a binary 
exposure and $y_i \in \{0,1\}$ is a binary outcome. We have $p = 3$ with $\mathfrak{F}_1(\data,s)(x)$ predicting $\P(A=1 \mid X=x)$, $\mathfrak{F}_2(\data,s)(x)$ predicting $\P(Y=1 \mid A=1, X=x)$, and $\mathfrak{F}_3(\data,s)(x)$ predicting $\P(Y=1 \mid A=0, X=x)$. We write  $\eta_s = \{\pi_s, \mu_s^1, \mu_s^0\}$, where $\pi_s = \{\mathfrak{F}_1(\data, s)(x_i) : z_i \in \data\}$, $\mu_s^1 = \{\mathfrak{F}_2(\data, s)(x_i) : z_i \in \data\}$, and $\mu_s^0 = \{\mathfrak{F}_3(\data, s)(x_i) : z_i \in \data\}$. The AIPW estimator of the ATE is
\begin{equation*}\label{eq:intercept}
    \Psi(\data, \eta_s) = \frac{1}{n} \sum_{i=1}^n \left(\frac{a_i}{\pi_{s,i}} - \frac{1 - a_i}{1 - \pi_{s,i}} \right)\left(y_i - \mu_{s,i}^{a_i}\right) + \mu_{s,i}^{1} - \mu_{s,i}^{0}.
\end{equation*}
Conditional on $\data$, the only source of randomness in $\Psi(\data, \eta_s)$ is from seed dependence in the training of the models for generating nuisance predictions $\eta_s$. While Theorem~\ref{th:stability} guarantees seed $(\epsilon,\delta)$-stability for each set of the component predictions separately, it does not automatically imply that $\Psi$ is seed $(\epsilon,\delta)$-stable on $\data$. Furthermore, such estimators often require the use of a sample splitting procedure called \textit{cross-fitting} to avoid imposing a condition that the space of functions $\mathfrak{F}$ searches over is $P$-Donsker \citep{Zheng2011,chernozhukov2018double}, which is typically unjustified if machine learning algorithms are used. The standard cross-fitting procedure uses $V$-fold cross-validation. Although technically only 2 folds are required to eliminate the Donsker condition, how to choose the number of folds to optimize finite-sample performance remains an open question. Leave-one-out (LOO) cross-fitting is guaranteed to eliminate the variability induced by sample splitting, but there is no guarantee it sufficiently controls the remaining variability due to seed stochasticity in the nuisance models themselves. Moreover, the standard proof that a DML estimator with cross-fitting is asymptotically normal without a Donsker condition does not extend to LOO cross-fitting, leaving the asymptotic normality of such estimators without theoretical justification.

Denote bagged nuisance predictions $\bar\eta_\zeta = \{\bar\eta_{\zeta,l}\}_{l=1}^p$ where $\bar\eta_{\zeta,l} = \{\bar{\mathfrak F}_l (\data, \rho, \zeta)(x_i) : z_i \in \data\}$ is a vector of $n$ predictions from the $l$th bagged prediction function generated under seed and bag pairs $\zeta$. The following corollary establishes a sufficient condition on $V$ for the function $\Psi(\data, \bar\eta_\zeta)$ to be seed $(\epsilon,\delta)$-stable on $\data$.

\begin{corollary}\label{cor:psistab}
Suppose the conditions of Theorem~\ref{th:stability} hold and function $\Psi(\data,\eta)\in \mathds{R}$ is Lipschitz continuous in $\eta$ with Lipschitz constant $L$ with respect to $\|\cdot\|_\infty$.
Then $\Psi(\data, \bar\eta_\zeta)$ is random seed $(\epsilon,\delta)$-stable on $\data$ whenever
\[V \geq \max_{z_i \in \data, l \in [p]}\frac{\log(2np/\delta)}{(\epsilon/L)^2}\left(4\nu^2_l(x_i;\data,\rho) + \frac{2}{3}\frac{\epsilon}{L}\right),\]
where the point-wise prediction variance over bags and seeds $\nu_l^2(x;\data,\rho)$ for the $l$th prediction function is defined in Theorem~\ref{th:stability}.
\end{corollary}
The Lipschitz condition will hold whenever all of the partial derivatives of $\Psi(\data, \eta)$ with respect to estimated nuisance parameters $\eta$ exist and are bounded. For example, it trivially holds for the AIPW estimator of the ATE when $c < \pi_i < 1-c$ for some $c > 0$. 
However, there are a variety of conditions under which the Lipschitz condition could hold.

Existing approaches have targeted seed stability by averaging over final estimates \citep{chernozhukov2018double,schader2024don}, on the grounds that averaging over nuisance parameters would not account for variance induced by sample splitting, and that repeated cross-fitting would be prohibitively computationally expensive \citep{schader2024don}. We propose a new cross-fitting procedure, cross-bagging, that simultaneously eliminates seed dependence arising from both machine learning estimation of nuisance parameters and sample splitting. The algorithm is as follows:

\begin{enumerate}
	\item Fix $\rho \in (0, 1)$ and choose $V$ large enough to satisfy Corollary~\ref{cor:psistab}.
	\item Construct $\zeta$: Draw $m = \lfloor \rho n\rfloor$-sized subsamples of indices from $\Pb$ until each observation is OOB at least $V$ times. Denote the total number of subsamples as $V^\dagger$. Further draw $V^\dagger$ seeds from $\Ps$ to pair with each bag.
	\item For $(v,l) \in [V^\dagger] \times [p]$, train $\mathfrak{F}_l(\data(B_m^v), s_v)$.
	\item Define matrix $M \in \{0,1\}^{n \times V^\dagger}$ where $M_{iv} = \mathds{1}[i \notin B_m^v]$, and $C = M \mathbf{1}_{V^\dagger}$. For $(z_i,l) \in \data \times [p]$, pool OOB nuisance estimates:
	$$
	  \bar{\mathfrak{F}}_l(\data, \rho,\zeta)(x_i) = \frac{1}{C_i} \sum_{v=1}^{V^\dagger} M_{iv}\,\mathfrak{F}_l(\data(B_m^v), s_v)(x_i).
	$$
	\item 
	Let $\bar\eta_\zeta = \{\bar\eta_{\zeta,l}\}_{l=1}^p$ where $\bar\eta_{\zeta,l} = \{\bar{\mathfrak F}_l (\data, \rho, \zeta)(x_i) : z_i \in \data\}$. The cross-bagged estimator is
	$\Psi(\data,\bar\eta_\zeta)$.
\end{enumerate}

The practical utility of this algorithm is limited as it either requires knowledge of $\nu_l^2(x_i;\data,\rho)$ or assuming the upper bound for $\nu_l^2(x_i;\data,\rho)$. Further, while the bound on $V$ in Corollary~\ref{cor:psistab} is sufficient for seed $(\epsilon,\delta)$-stability, it is not necessary; it may be prohibitively large and we expect the required number of bags to be substantially smaller in practice. In principle, this can be iterated: run the algorithm on many seeds with an initial $V$, estimate the probability that two estimates differ by more than $\epsilon$, and increase the number of bags until this probability falls below $\delta$. However, this approach is computationally expensive. Instead, progress can be made by quantifying the seed variability and simultaneously incrementing both the number of seeds and the number of bags while recycling previous predictions until the seed stability condition is met.

\begin{theorem}\label{th:mcvar} 
Assume $\Psi(\data,\eta)$ is differentiable with respect to $\eta$. Let $\zeta$ and $\zeta'$ be sets of independent draws of random seed-and-bag pairs from $\Ps \otimes \Pb$, and $\bar \eta_\zeta$ and $\bar \eta_{\zeta'}$ be nuisance predictions obtained via cross-bagging. Then as $V \rightarrow \infty$
    \[
    V^{1/2}\left(\Psi(\data,\bar\eta_\zeta) - \Psi(\data,\bar\eta_{\zeta'})\right)
    \rightsquigarrow \mathcal{N}(0, 2\tau^2),
    \]
    where $\tau^2$ is the asymptotic variance of $V^{1/2} (\Psi(\data,\bar\eta_\zeta) - \Psi(\data,\bar\eta^\infty))$ as $V \rightarrow \infty$ and can be consistently estimated by
\[
\hat\tau^2 = \frac{1}{V(1 - m/n)^2}\sum_{v=1}^V \left(\sum_{l=1}^p \left[\sum_{i=1}^n \frac{\partial \Psi(\data, \bar\eta_\zeta)}{\partial \bar\eta_{\zeta,l,i}} M_{iv}\left(\,\mathfrak{F}(\data(B_m^v),s_v)(x_i) - \bar\eta_{\zeta,l,i}\right)\right]\right)^2 
\]
\end{theorem}

Theorem~\ref{th:mcvar} is an asymptotic statement about the number of bags, not the sample size used for training the nuisance models. Equipped with $\hat{\tau}^2$, we may adaptively cross-bag until a desired level of seed stability is reached.

\begin{enumerate}
    \item Fix $\rho \in (0,1)$ and an initial target OOB threshold $V > 1$. Let $m = \lfloor \rho n \rfloor$. Construct $\zeta$, train nuisance models, compute pooled OOB predictions $\bar\eta_\zeta$, and form the cross-bagged estimator $\Psi(\data, \bar\eta_\zeta)$ as in the cross-bagging procedure above.
    \item Estimate $\hat\tau^2$ using the variance estimator in Theorem~\ref{th:mcvar}.
    \item Let $t_{V^\dagger - 1,1-\delta/2}$ be the $1-\delta/2$ quantile of the $t$-distribution with $V^\dagger - 1$ degrees of freedom. If
    \[
    t_{V^\dagger - 1,1-\delta/2}\left(\frac{2\hat\tau^2}{V^\dagger}\right)^{1/2} \leq \epsilon,
    \]
    return $\Psi(\data, \bar\eta_\zeta)$. Otherwise, draw additional seed-and-bag pairs, append them to $\zeta$, update $V^\dagger$, recompute the pooled OOB predictions $\bar\eta_\zeta$ and the estimator $\Psi(\data, \bar\eta_\zeta)$, and return to step 2.
\end{enumerate}

The adaptive cross-bagging algorithm constructs a $1-\delta$ confidence interval around $\E_{\zeta,\zeta'}[\Psi(\data,\bar\eta_\zeta) - \Psi(\data,\bar\eta_{\zeta'})] = 0$ and iteratively checks that this interval is contained in $[-\epsilon,\epsilon]$. The gradient in Theorem~\ref{th:mcvar} can easily be calculated using standard automatic differentiation software without deriving the analytic form.

\section{Numerical experiments}\label{sec:sim}

The primary aim of our numerical simulation is to (i) show that bagging stabilizes a seed-unstable algorithm $\mathfrak{F}$ given sufficiently many bags $V$, and (ii) compare our adaptive cross-bagging algorithm for controlling $(\epsilon,\delta)$-stability of $\Psi(\data,\eta)$ against na\"{i}ve averaging and increased cross-fitting folds. 

For the first goal, we generate a random data generating process (DGP) with one binary outcome $Y$ and 20 features $X$ using Claude Opus 4.6 \citep{anthropic2025claude}; the DGP may be found in the supplementary materials. We then simulate one dataset ($n=100$) and a single test observation, and draw $1000$ seeds from $\text{Unif}(\M)$. For each seed, we train a single layer neural network with 20 neurons \citep{rumelhart1986learning} and a subbagged version of the neural network for $\P(Y = 1 \mid X)$; we generate a prediction from each model on the same test observation. We use $\rho = 2/3$ and $V=320$ for training the subbagged neural network, where the number of bags is chosen to guarantee seed $(0.1,0.1)$-stability under the worst case variance of $\nu^2(x;\data,\rho) = 1/4$.

For the second aim we use the DGP from \citet{schader2024don}; it consists of four confounders $X$, binary exposure $A$, and binary outcome $Y$ and may be found in the supplementary materials. The target parameter $\psi_0$ is the ATE of $A$ on $Y$ adjusting for $X$. Again, we simulate a single dataset ($n=100$) and draw $1000$ seeds. For each seed, we estimate the ATE using AIPW with (i) 2-fold, 10-fold, and LOO cross-fitting, (ii) the ``averaging on final estimate'' approach from \citet{schader2024don} and \citet{chernozhukov2018double} with 80 seeds and 2-fold (the parameterization used by \citet{schader2024don}), 10-fold, and LOO cross-fitting, and (iii) with adaptive cross-bagging ($\rho = 1/2$) and stability criterion: $(\epsilon, \delta) = (0.01,0.01)$. Replicating \citet{schader2024don}, we estimate nuisance parameters using a random forest with 500 trees \citep{breiman2001random}. 

To avoid cherry-picking the generated datasets, an initial seed was also generated with Claude Opus 4.6 \citep{anthropic2025claude}. All simulations were conducted in \texttt{R} \citep{R}. Simulation code is available at \url{https://github.com/nt-williams/bagging-random-seed-stability}. 

\begin{figure}
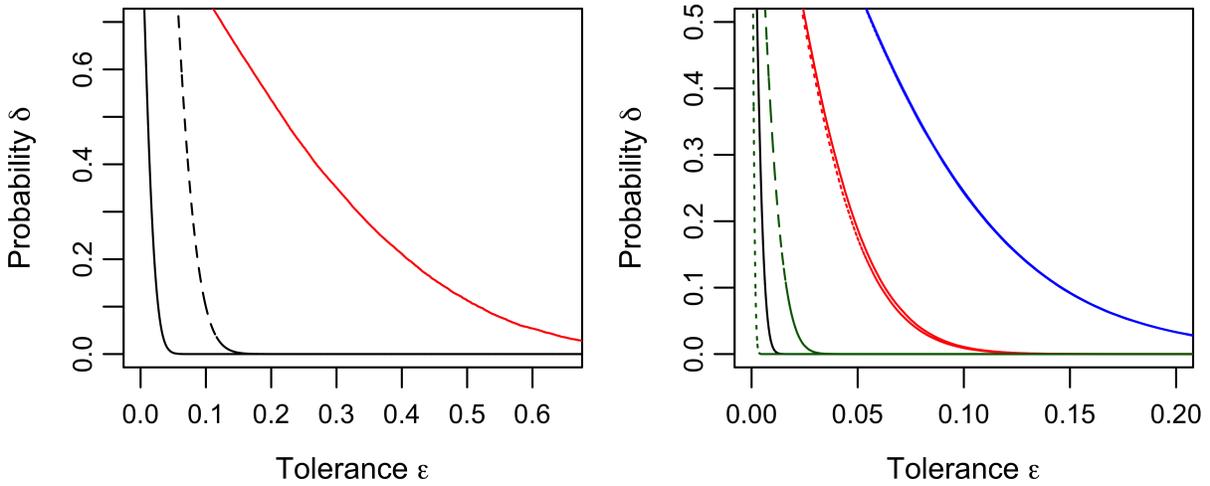

\centering
\begin{minipage}[t]{0.45\textwidth}
\centering
\figuresize{.125}
\figurebox{2pc}{2pc}{}[nnet-phase]
\end{minipage}%
\hspace{0.1\textwidth}%
\begin{minipage}[t]{0.45\textwidth}
\centering
\figuresize{.125}
\figurebox{2pc}{2pc}{}[adaptive-crossbag-aipw]
\end{minipage}
\caption{Empirical seed $(\epsilon,\delta)$-stability. Left: single layer neural network predictions. The red line corresponds to no bagging and the black line to subbagging using the estimated minimum $V$-bags for seed $(0.1,0.1)$-stability (dashed line). Right: AIPW estimator for the ATE. The blue, red, and green lines correspond to 2-fold, 10-fold, and LOO cross-fitting respectively; the black line corresponds to adaptive cross-bagging for seed $(0.01,0.01)$-stability. Dashed lines correspond to no seed averaging, while dotted lines correspond to averaging over 80 seeds.}
\label{fig:phase}
\end{figure}

\subsection{Results}

Results of the numerical experiments are shown in Figure~\ref{fig:phase}, where we plot an empirical estimate of the seed stability criterion averaged over a large number of seeds
\[\hat\delta_\epsilon = \frac{1}{\binom{1000}{2}} \sum_{i < j} \mathds{1}\!\left[\left|p_i - p_j\right| > \epsilon\right].\]
We assess random-seed stability via $r(\hat\delta_\epsilon, \delta) = \hat\delta_\epsilon / \delta$, where $r \leq 1$ indicates the $(\epsilon,\delta)$-stability criterion is met and larger $r$ indicates looser control (e.g., $r=10$ indicates the empirical stability was 10 times the tolerated level). In simulation 1, the subbagged neural network ($V=320$) achieves conservative stability, $r(\hat\delta_{0.1}, 0.1) = 0$, whereas the neural network without bagging yields $r(\hat\delta_{0.1},0.1) = 7.5$. In simulation 2, AIPW with adaptive cross-bagging controls stability near the target: $r(\hat\delta_{0.01},0.01) = 1.04$. Without averaging, 2-fold cross-fitting gives $r(\hat\delta_{0.01}, 0.01) = 90.5$; 10-fold, $79.12$; LOO, $31.52$. Averaging over 80 seeds, these become $90.35$ (2-fold), $78.33$ (10-fold), and $0$ (LOO). The median number of bags selected by the adaptive cross-bagging algorithm was $396$.

On an Apple M3 Max with 36 GB RAM, the adaptive cross-bagging algorithm achieved $(0.01,0.01)$-stability in approximately 20 seconds. The averaging on final estimate approach ran in approximately 11 seconds (2-folds), 1 minute (10-fold), and 10 minutes (LOO). While the averaging on final estimate approach with LOO cross-fitting also achieved $(0.01,0.01)$-stability, it required roughly 20 times the computation time.

\section*{Declaration of the use of generative AI and AI-assisted technologies}

During the preparation of this work the authors used Claude Opus 4.6 \citep{anthropic2025claude} in order to generate a data generating process for numerical simulation and to improve clarity of writing. After using this tool/service the author(s) reviewed and edited the content as necessary and take full responsibility for the content of the publication.

\section*{Supplementary material}
\label{SM}
The Supplementary Material includes proofs of the results and data generating processes.

\bibliographystyle{biometrika}
\bibliography{paper-ref}

\clearpage
\appendix

\section{Proofs}\label{appn:proofs}
\subsection*{Lemma 1}

\begin{proof}
    Define $R(x) := \mathfrak{F}(\data, s)(x) - \mathfrak{F}(\data, s')(x)$, which is a random variable in $[-1,1]$.
Because $\mathfrak{F}(\data, s)(x)$ and $\mathfrak{F}(\data, s')(x)$ are independent draws from $\mathcal F(\data)$
\begin{align*}
\text{Var}_{s,s'}\left(R(x) \right) &= 2\sigma^2(x;\data).
\end{align*}
In addition, recognize that $\mathfrak{F}(\data, s)(x)$ and $\mathfrak{F}(\data, s')(x)$ are unbiased for ${\mathfrak F}^\infty(\data)(x)$ such that $\E_{s,s'}[R(x)] = 0$. Then, by Bernstein's inequality 
\begin{align*}
\P_{s,s'}\left(\left| R(x) \right| \geq \epsilon \right) \leq 2 \exp\left(-\frac{\epsilon^2}{4\sigma^2(x;\data) + \frac{2}{3}\epsilon}\right).
\end{align*}
Let $k = \left|\test\right|$. We have
\begin{align*}
    \P_{s,s'}\left(\max_{x \in \test}\left|R(x)\right| \geq \epsilon\right) 
    &= \P_{s,s'}\left(\bigcup_{i=1}^k \left\{ \left|R(x_i)\right| \geq \epsilon\right\}\right) \\
    &\leq \sum_{i=1}^k \P_{s,s'}\left(\left|R(x_i)\right| \geq \epsilon\right)\\
    &\leq k \max_{x\in \test}\P_{s,s'}\left(\left|R(x)\right| \geq \epsilon\right).
\end{align*}
Therefore, to control $\P_{s,s'}\left(\max_{x \in \test}\left|R(x)\right| \geq \epsilon\right)$ with probability $\delta$, it suffices to control all point-wise predictions with probability $\delta/k$. Setting the right-hand side of the Bernstein inequality to less than or equal to $\delta/k$ and solving for $\sigma^2(\data,\test) = \max_{x \in \test} \sigma^2(x;\data)$ yields
\begin{align*}
\sigma^2(\data,\test) &\leq \frac{\epsilon}{4} \left(\frac{\epsilon}{\log(2k/\delta)} - \frac{2}{3}\right).
\end{align*}
\end{proof}

\subsection*{Theorem 1}

\begin{proof}
   Recall that $m = \lfloor\rho n\rfloor$ for $\rho \in (0,1)$, and that $(s,B) \sim \Ps \otimes \Pb$. Define
\[
\nu^2(x;\data,\rho) = \text{Var}_{(s,B)}(\mathfrak{F}(\data(B_m^v),s_v)(x)).
\]
Define $S(x) := \bar{\mathfrak{F}}(\data,\rho,\zeta)(x) -  \bar{\mathfrak{F}}(\data,\rho,\zeta')(x)$. Recognize that $E_{\zeta,\zeta'}[S(x)] = 0$, $\text{Var}_{\zeta,\zeta'}(S(x)) = 2 \nu^2(x;\data,\rho)/V$, and $S(x) \in [-1,1]$. Therefore, by Bernstein's inequality
\[
\P_{\zeta,\zeta'}\left(\left|S(x)\right| \geq \epsilon \right) \leq 2 \exp\left(-\frac{V\epsilon^2}{4\nu^2(x;\data,\rho) + \frac{2}{3}\epsilon}\right).
\]
Again, this point-wise bound can be made uniform over finite $\mathcal{T}\subset \mathcal X$, $k = |\mathcal{T}|$ with a union bound. Setting the right-hand side of the inequality less than or equal to $\delta/k$ and solving for $V$ yields
$$
V \ge \frac{\log(2k/\delta)}{\epsilon^2}\left(4\nu^2(x; \data,\rho) + \frac{2}{3}\,\epsilon\right),
$$
which holds for all $x \in \mathcal{T}$ if
$$
V \ge \max_{x\in \mathcal{T}}\frac{\log(2k/\delta)}{\epsilon^2}\left(4\nu^2(x; \data,\rho) + \frac{2}{3}\,\epsilon\right).
$$ 
\end{proof}

\subsection*{Corollary 1}

\begin{proof}
    Denote bagged nuisance predictions $\bar\eta_\zeta = \{\bar\eta_{\zeta,l}\}_{l=1}^p$ where $\bar\eta_{\zeta,l} = \{\bar{\mathfrak F}_l (\data, \rho, \zeta)(x_i) : z_i \in \data\}$ is a vector of $n$ bagged predictions from the $l$th regression algorithm generated under seed and bag pairs $\zeta$. Let $\left\|\bar\eta_\zeta \right\|_\infty = \max_i(\max_l(|\bar{\mathfrak F}_l (\data, \rho, \zeta)(x_i)|))$. Suppose that $\Psi(\data, \eta)$ is Lipschitz continuous in $\eta$ so that there is a constant $L$ such that
\[|\Psi(\data, \bar\eta_\zeta) - \Psi(\data,\bar\eta_\zeta')| \leq L \left\|\bar\eta_\zeta - \bar\eta_\zeta' \right\|_\infty.\]
Then, 
\begin{align*}
   \P_{\zeta,\zeta'}\left(\left|\Psi(\data, \bar\eta_\zeta) - \Psi(\data, \bar\eta_\zeta') \right|\geq \epsilon \right) &\leq \P_{\zeta,\zeta'}\left(L\left\|\bar\eta_\zeta - \bar\eta_\zeta' \right\|_\infty \geq \epsilon  \right) \\
   &=\P_{\zeta,\zeta'}\left(\left\|\bar\eta_\zeta - \bar\eta_\zeta' \right\|_\infty \geq \frac \epsilon L \right)\\
   &\leq \sum_{i=1}^n \sum_{l=1}^p \P_{\zeta,\zeta'} \left(\left|\bar\eta_{\zeta,l,i} - \bar\eta_{\zeta',l,i}\right| \geq \frac \epsilon L \right),
\end{align*}
Therefore to control $\P_{\zeta,\zeta'}\left(\left\|\bar\eta_\zeta - \bar\eta_\zeta' \right\|_\infty \geq \frac \epsilon L \right)$ with probability $\delta$, it suffices to control the random seed stability of all point-wise bagged predictions on $\data$ with probability $\delta/(np)$. 
From the previous proof we have
\[
\P_{\zeta,\zeta'} \left(\left|\bar\eta_{\zeta,l,i} - \bar\eta_{\zeta',l,i}\right| \geq \frac \epsilon L \right) \leq 2 \exp\left(-\frac{V\left(\epsilon/L\right)^2}{4\nu_l^2(x_i;\data,\rho) + \frac{2}{3}\left(\epsilon/L\right)}\right).
\]
Setting the right-hand side less than or equal to $\delta/(np)$ and solving for $V$ 
\[
V\geq \frac{\log(2np/\delta)}{(\epsilon/L)^2}\left(4\nu^2_l(x_i;\data,\rho) + \frac{2}{3}\frac{\epsilon}{L}\right),
\]
which holds simultaneously when
\[
V\geq \max_{l \in [p],z_i\in \data}\frac{\log(2np/\delta)}{(\epsilon/L)^2}\left(4\nu^2_l(x_i;\data,\rho) + \frac{2}{3}\frac{\epsilon}{L}\right).
\]
\end{proof}

\subsection*{Theorem 2}

\begin{proof}
    Let $\mathbf{1}_V$ be a vector of ones with length $V$. Define matrix $M \in \{0,1\}^{n \times V}$ where $M_{iv} = \mathds{1}[i \notin B_m^v]$, and $C = M \mathbf{1}_V$. Cross-bagged predictions of algorithm $\mathfrak F$ on data $\data$ are then given by
\[
\bar{\mathfrak F}(\data, \rho, \zeta)(x_i) = \frac{1}{C_i} \sum_v M_{i,v}\,\mathfrak{F}(\data(B_m^v),s_v)(x_i) = \frac{1/V \sum_vM_{i,v}\,\mathfrak{F}(\data(B_m^v),s_v)(x_i)}{C_i/V},
\]
The numerator is a sample average of i.i.d terms and therefore the central limit theorem may be applied. Note that as $V\to \infty$, $C_i/V \overset{p}{\rightarrow} (1 - m/n) \neq 0$. Therefore, by Slutsky's theorem, $\sqrt{V}\left(\bar{\mathfrak F}(\data, \rho, \zeta)(x_i) - \bar{\mathfrak{F}}^\infty(\data,\rho)(x_i)\right)$ also converges to a normal distribution as $V \to \infty$.
Define 
\[
\bar\eta_\zeta = \left\{\bar{\mathfrak F}(\data, \rho, \zeta)(x_i) : z_i \in \data \right\},\quad\bar\eta^\infty = \left\{\bar{\mathfrak F}^\infty(\data, \rho)(x_i) : z_i \in \data \right\}
\]
Then by the multivariate central limit theorem, 
\[
\sqrt{V}\left(\bar\eta_\zeta - \bar\eta^\infty\right) \rightsquigarrow  \mathcal{N}(0, \Sigma).
\]
Suppose the function $\Psi(\data, \eta)$ is differentiable in $\eta$. Then by the delta method, 
\[
\sqrt{V}\left(\Psi(\data,\bar\eta_\zeta) - \Psi(\data,\bar\eta^\infty)\right) \rightsquigarrow \mathcal{N}(0, \tau^2)
\]
where the asymptotic variance $\tau^2 = \nabla_\eta \Psi(\data, \bar\eta^\infty)'\Sigma \nabla_\eta \Psi(\data, \bar\eta^\infty)$. Recognize that $\tau^2$ is the asymptotic variance of the first order Taylor expansion $V^{1/2} (\Psi(\data, \bar\eta_\zeta) - \Psi(\data, \bar\eta^\infty)) \approx V^{1/2} \nabla_\eta \Psi(\data, \bar\eta^\infty)'(\bar\eta_\zeta - \bar\eta^\infty)$. Expressed algebraically, we have:
\begin{align*}
    \nabla_\eta \Psi(\data, \bar\eta^\infty)'(\bar\eta_\zeta - \bar\eta^\infty)
    &= \sum_i \frac{\partial \Psi(\data, \bar\eta^\infty)}{\partial \bar\eta_i^\infty}(\bar\eta_{\zeta,i} - \bar\eta^\infty_i)\\
    &= \sum_i \frac{\partial \Psi(\data, \bar\eta^\infty)}{\partial \bar\eta_i^\infty} \left(\frac{1}{C_i} \sum_v M_{i,v}\,\mathfrak{F}(\data(B_m^v),s_v)(x_i) -  \bar\eta^\infty_i\right)\\
    &= \sum_i \frac{\partial \Psi(\data, \bar\eta^\infty)}{\partial \bar\eta_i^\infty} \left(\frac{1}{C_i} \sum_v M_{i,v}\,\mathfrak{F}(\data(B_m^v),s_v)(x_i) -  \frac{1}{C_i} \sum_v M_{i,v}\,\bar\eta^\infty_i\right)\\
    &= \sum_v \sum_i \frac{\partial \Psi(\data, \bar\eta^\infty)}{\partial \bar\eta_i^\infty} \frac{M_{i,v}}{C_i}\left(\,\mathfrak{F}(\data(B_m^v),s_v)(x_i) - \bar\eta^\infty_i\right)\\
    &\approx \sum_v \sum_i \frac{\partial \Psi(\data, \bar\eta^\infty)}{\partial \bar\eta_{i}^\infty} \frac{M_{i,v}}{V(1- m/n)}\left(\,\mathfrak{F}(\data(B_m^v),s_v)(x_i) - \bar\eta^\infty_i\right)\\
    &= \frac{1}{V}\sum_v \frac{1}{(1 - m/n)} \sum_i \frac{\partial \Psi(\data, \bar\eta^\infty)}{\partial \bar\eta_{i}^\infty} M_{i,v}\left(\,\mathfrak{F}(\data(B_m^v),s_v)(x_i) - \bar\eta^\infty_i\right),
\end{align*}
where the second approximation uses the fact that as $V\to \infty$, $C_i/V \overset{p}{\rightarrow} (1 - m/n)$. Denote
\[
U(M_{:,v}, \bar\eta_\zeta) = \frac{1}{(1 - m/n)} \sum_i \frac{\partial \Psi(\data, \bar\eta_\zeta)}{\partial \bar\eta_{\zeta,i}} M_{i,v}\left(\,\mathfrak{F}(\data(B_m^v),s_v)(x_i) - \bar\eta_{\zeta,i}\right).
\]
Because $\bar\eta_\zeta \overset{p}{\rightarrow} \bar\eta^\infty$ as $V \rightarrow \infty$ and by the continuous mapping theorem $U(M_{:,v}, \bar\eta_\zeta)^2 \overset{p}{\rightarrow} U(M_{:,v},\bar\eta^\infty)^2$ as $V \rightarrow \infty$. Then by the weak law of large numbers and Slutksy's lemma
\[
\frac{1}{V}\sum_v U(M_{:,v}, \bar\eta_\zeta)^2 \overset{p}{\rightarrow} E[U(M_{:,v},\bar\eta^\infty)^2].
\]
Using the fact that $E[U(M_{:,v},\bar\eta^\infty)] = 0$, a consistent estimator for $\tau^2$ is then
\[
\hat\tau^2 = \frac{1}{V(1 - m/n)^2}\sum_v \left(\sum_i \frac{\partial \Psi(\data, \bar\eta_\zeta)}{\partial \bar\eta_{\zeta,i}} M_{i,v}\left(\,\mathfrak{F}(\data(B_m^v),s_v)(x_i) - \bar\eta_{\zeta,i}\right) \right)^2, 
\]
which generalizes to the case where $\eta$ has $p > 1$ dimensions:
\[
\hat\tau^2 = \frac{1}{V(1 - m/n)^2}\sum_v \left(\sum_l \left[\sum_i \frac{\partial \Psi(\data, \bar\eta_\zeta)}{\partial \bar\eta_{\zeta,l,i}} M_{i,v}\left(\,\mathfrak{F}(\data(B_m^v),s_v)(x_i) - \bar\eta_{\zeta,l,i}\right)\right] \right)^2.
\]
\end{proof}

\subsection*{Miscellanea results}

``For example, it trivially holds for the AIPW estimator of the ATE when $c < \pi_i < 1-c$ for some $c > 0$. ''

\begin{proof}
We have
\begin{equation*}\label{eq:intercept}
    \Psi(\data, \eta) = \frac{1}{n} \sum_{i=1}^n \left(\frac{a_i}{\pi_i} - \frac{1 - a_i}{1 - \pi_i} \right)\left(y_i - \mu_{i}^{a_i}\right) + \mu_{i}^{1} - \mu_{i}^{0}.
\end{equation*}
Which has partial derivatives with respect to $\eta$:
\begin{align*}
    \frac{\partial \Psi(\data,\eta)}{\partial \pi_i} &= -\frac{1}{n}\frac{2a_i -1}{\pi_i^2}(y_i - \mu_i^{a_i})\\
    \frac{\partial \Psi(\data,\eta)}{\partial \mu_i^0} &= \frac{1}{n}\left(\frac{1-a_i}{\pi_i} - 1\right)\\
    \frac{\partial \Psi(\data,\eta)}{\partial \mu_i^1} &= \frac{1}{n}\left( 1 - \frac{a_i}{\pi_i}\right) 
\end{align*}
A sufficient condition for all partial derivatives to be bounded:
\begin{enumerate}
    \item $\pi_i > c > 0$, 
    \item boundedness of $\mu_i^a$ for $a \in \{0,1\}$, and
    \item boundedness of $y_i$
\end{enumerate}
Conditions 2 and 3 hold by construction.
\end{proof}

\section{Simulation details}\label{app:dgps}
Define $\mathbf{0}_k$ and $\mathbf{1}_k$ as vectors of zeros and ones with size $k$, respectively. For the first simulation, we generated data from the following data generating mechanism:
\begin{equation*}
\begin{split}
(W_1, \ldots, W_5)^\top &\sim \mathcal{N}\!\left(\mathbf{0}_5,\; \Sigma_1\right), \quad [\Sigma_1]_{j,k} = 0.3 \cdot \mathds{1}(j \neq k) + \mathds{1}(j = k) \\
(W_6, \ldots, W_{10})^\top &\sim \mathcal{N}\!\left(\mathbf{1}_5,\; \Sigma_2\right), \quad [\Sigma_2]_{j,k} = 0.2 \cdot \mathds{1}(j \neq k) + \mathds{1}(j = k) \\
W_{11}, \ldots, W_{15} &\overset{\text{iid}}{\sim} \mathcal{N}(0,\; 1.5^2) \\
W_{16} &\sim \text{Gamma}(2,\; 1) \\
W_{17} &= 4(B - 1), \quad B \sim \text{Beta}(2,\; 5) \\
W_{18} &\sim \mathcal{N}(0,\; 4) \\
W_{19} &\sim t(\text{df} = 5) \\
W_{20} &\sim \text{Uniform}(-2,\; 2) \\[10pt]
\text{logit}\; E[Y \mid W] &= -0.5 \\
  &\quad + 0.5\, W_1 - 0.3\, W_2^2 + \sin\!\left(\tfrac{\pi}{2} W_3\right) + 0.4\, (W_4)_+ - 0.6\, |W_5| \\
  &\quad + 0.2\, W_6 + 5\, \phi(W_7;\, 1,\, 0.5) - 0.3\, W_8 \\
  &\quad + 0.5\, \text{sign}(W_9)\log(|W_9| + 1) - 0.15\, W_{10} \\
  &\quad + 0.8\, \mathds{1}(W_{11} > 1)\, W_{11} - 0.04\, W_{13}^3 \\
  &\quad + 0.3\, \sqrt{(W_{16})_+} - 0.6\cos(W_{18}) + 0.2\, W_{20} \\
  &\quad + 0.5\, W_1 W_6 - 0.4\, W_2 W_3 \\
  &\quad + 0.6\, W_4 \cdot \mathds{1}(W_7 > 0.5) - 0.3\, W_5 \sin(W_{18}) \\
  &\quad + \frac{0.7\, W_{11}\, W_{16}}{1 + |W_{11}|} - 0.5\, \mathds{1}(W_8 > 0)\, \mathds{1}(W_{13} > 0) \\
  &\quad + 0.4\, W_1\, W_7 \cdot \mathds{1}(W_{20} > 0) \\[10pt]
Y &= \text{expit}\!\left(\text{logit}\; \E[Y \mid W]\right)
\end{split}
\end{equation*}

For the second simulation we used the following data generating mechanism from \citet{schader2024don}:

\begin{equation*}
\begin{split}
W_1 &\sim \text{Uniform}(0,\; 2) \\
W_2,\, W_3,\, W_4 &\overset{\text{iid}}{\sim} \text{Bernoulli}(0.5) \\
\text{logit}\; {\rm pr}(A = 1 \mid W) &= W_1 + W_2 W_3 - 2\, W_4 \\
\text{logit}\; {\rm pr}(Y = 1 \mid A, W) &= W_1 + W_2 W_3 - 3.
\end{split}
\end{equation*}

\end{document}